\documentclass[12pt]{amsart}

\usepackage[pdftex]{graphicx}
\usepackage{amssymb}
\usepackage{color}

\usepackage[hmargin=2.5cm,vmargin=2.5cm]{geometry}
\usepackage{enumerate}
\usepackage{subfigure}

\DeclareMathOperator{\supp}{supp}
\DeclareMathOperator{\arcsinh}{arcsinh}

\newcommand{\Z}{\mathbb{Z}}

\renewcommand{\kappa}{\varkappa}

\theoremstyle{plain}
\newtheorem{thm}{Theorem}[section]

\newtheorem{lemma}[thm]{Lemma}

\theoremstyle{definition}
\newtheorem{remark}[thm]{Remark}
\newtheorem{definition}[thm]{Definition}
\newtheorem{example}[thm]{Example}

\numberwithin{equation}{section}

\title{Simplicity of eigenvalues and non-vanishing of eigenfunctions
  of a quantum graph}

\author{Gregory Berkolaiko}
\author{Wen Liu}
\address{Department of Mathematics, Texas A\&M University, College Station, TX 77843-3368, USA}

\begin{document}

\newcommand\todo[1]{{{\bf To do:} \color{red}#1}}

\maketitle

\begin{abstract}
  We prove that after an arbitrarily small adjustment of edge lengths,
  the spectrum of a compact quantum graph with $\delta$-type vertex
  conditions can be simple.  We also show that the eigenfunctions,
  with the exception of those living entirely on a looping edge, can
  be made to be non-vanishing on all vertices of the graph.

  As an application of the above result, we establish that the secular
  manifold (also called ``determinant manifold'') of a large family of
  graphs has exactly two smooth connected components.
\end{abstract}

\section{Introduction}

A quantum graph is a metric graph equipped with a self-adjoint
differential operator (usually of Schr\"odinger type) defined on the
edges and matching conditions specified at the vertices. Every edge of
the graph has a length assigned to it.

One of the fundamental questions of the spectral theory is that of
presence in the spectrum of degenerate (or repeated) eigenvalues.  In
particular, it is usually the case that within a rich enough set of
problems, the problems \emph{with} degenerate eigenvalues form a small
subset.  In other words, unless a system has symmetries (which usually
force degeneracy in the spectrum, see, for example,
\cite{Wigner_grouptheory}), it is highly unlikely to have degenerate
eigenvalues.

Mathematically, a classical result by Uhlenbeck~\cite{Uhl_bams72} (see
also \cite{Uhl_ajm76} for a generalization) establishes generic
simplicity of eigenvalues of the Laplace-Beltrami operator on compact
manifolds, with respect to the set of all possible metrics on the
manifold.  Some generic properties of eigenfunctions are also
established. Since then, various extensions and generalizations of
this result have been proven for different circumstances (see, for
example, \cite{GuiLegSen_gsa14} and references therein).

On graphs, the question of simplicity of eigenvalues was considered by
Friedlander in \cite{Fri_ijm05}, who proved that the eigenvalues are
simple generically with respect to the perturbation of the edge
lengths of the graph.  The proof is based on perturbation theory and
applies to graphs with Neumann--Kirchhoff (NK) conditions only (see
Section~\ref{sec:definitions} for the definitions).  When this article
was in preparation, an outline of a shorter proof, under the same
conditions, was released by Colin de Verdi\`ere \cite{CdV_ahp15}.

In this work we consider a wider range of vertex conditions, namely
the $\delta$-type conditions on vertices of the graph.  Furthermore,
we also investigate the eigenfunctions, showing that generically they
do not vanish on vertices, unless this is unavoidable due to presence
of looping edges.  Both of these results are important in
applications, in particular all recent results on the number of zeros
of graph eigenfunctions assume both the simplicity of eigenvalues and
non-vanishing of eigenfunctions on vertices as a precondition (see
\cite{Ber_cmp08,BanOreSmi_incoll08,BanBerSmi_ahp12,Ban_ptrsa14,BerWey_ptrsa14}
and references therein).

In the proof, the simplicity of eigenvalues and non-vanishing of
eigenfunctions are tightly interconnected; each property is assisting in
the proof of the other (the proof is done by induction).  The proof is
geometric in nature and uses local modifications of the graph to reduce
it to previously considered case.  In Section~\ref{sec:secular} of the
paper we also consider an application of the result to the study of
the secular manifold of a graph, showing that for large classes of
graphs, the set of smooth points of the manifold has exactly two
connected components.

We remark that from the general consideration one can deduce the
result for generic choices of the vertex conditions.  The challenge is
to obtain it for a fixed choice of vertex conditions (and a generic
choice of edge lengths).  The existing proofs cannot be readily
re-used for this purpose.  While the original proof due to Friedlander
\cite{Fri_ijm05} is very technical, the simpler proof by Colin de
Verdi\`ere \cite{CdV_ahp15} relies on the properties of the so-called
``secular manifold'' for quantum graphs which does not exist for
general $\delta$-type conditions.  Finally, we mention a result of
Exner and Jex, where the simplicity of the ground state eigenvalue and
positivity of the corresponding eigenfunction was established for
graph with non-repulsive $\delta$-type conditions \cite{ExnJex_pla12}.

\section{Quantum graph Hamiltonian}
\label{sec:definitions}

We start by defining the quantum graph, following the notational
conventions of \cite{BerKuc_graphs}. Let $\Gamma=(V,E)$ be a connected
metric graph with a set of vertices $V=\{v_j\}$ and edges $E=\{e_j\}$.
Both sets $V$ and $E$ are assumed to be finite and the edges are of
bounded length.  We allow multiple edges between a given pair of
vertices and the edges that loop from a vertex to itself (see also
Remark~\ref{rem:vertex_deg2} below).

A function $f$ on $\Gamma$ is a collection of functions $f_e(x)$
defined on each edge $e$.  Consider the Laplace operator $H$ defined
by
\begin{equation*}
 H: f \mapsto -\dfrac{d^2 f}{dx^2},
\end{equation*}
acting on the functions that belong to the Sobolev $H^2(e)$ space on
each edge $e$ and satisfy the $\delta$-type boundary conditions with
coefficients $\alpha_v$ at the vertices of the graph,
\begin{equation}
\label{eq:cond_deltatype}
\begin{cases}
  f(x) \mbox{ is continuous at } v\\
  \sum_{e\in E_v}\dfrac{df}{dx_e}(v)=\alpha_vf(v),
\end{cases}
\end{equation}
where for each vertex $v$, the corresponding vertex condition
$\alpha_v$ is a fixed real number. The set $E_v$ is the set of edges
joined at the vertex $v$; by convention, each derivative at a vertex
is taken into the corresponding edge.  We will often encounter the
special case with $\alpha_v=0$, which is known as the
Neumann--Kirchhoff (NK) condition.  The special value
$\alpha_v=\infty$ should be taken to mean the Dirichlet condition
$f(v) = 0$.  Such condition will only be allowed at vertices of degree
1, as it effectively disconnects the edges if imposed at a vertex of
degree 2 or higher.  Conditions with $\alpha_v \neq 0, \infty$ will be
called \emph{Robin-type}.

\begin{remark}
  \label{rem:vertex_deg2}
  NK condition (equation (\ref{eq:cond_deltatype}) with $\alpha_v=0$)
  at a vertex of degree 2 is equivalent to $f$ being continuously
  differentiable at $v$.  Therefore, a graph with an NK vertex of
  degree 2 is equivalent to a graph which has no vertex at this
  location, just a continuous edge, see Fig.~\ref{fig:loop}.  We will
  often use this fact in reverse, choosing a point on an edge and
  declaring it to be a vertex of degree 2 with NK condition.  We will
  call such a vertex a \emph{trivial vertex}.

  Note that by introduction of such trivial vertices, a graph with
  multiple or looping edges may be converted into a \emph{simple}
  graph.
\end{remark}

\section{Main results}
\label{sec:results}

The question we address here is when is it typical (with respect to
variation of edge lengths) for a graph to have simple spectrum and to
have eigenfunctions that do not vanish on vertices of the graph.  To
motivate our results, we first consider examples which turn out to be
the only cases one needs to take special care about.

\begin{definition}
  \label{def:loop}
  A \emph{loop} is a chain of vertices
  \begin{equation*}
    v,\ v_1,\ \ldots \ v_n,\ v,
  \end{equation*}
  connected by edges, with each of the intermediate vertices
  $v_1,\ldots, v_n$ having degree 2.  We include the possibility of
  having $n=0$, in which case $v$ is connected to itself by a looping
  edge.
\end{definition}

By Remark~\ref{rem:vertex_deg2}, a looping edge is equivalent to a
loop with intermediate vertices with $\alpha_{v_j}=0$, see
Fig.~\ref{fig:loop}.

\begin{figure}
  \centering
  \includegraphics[width=10cm]{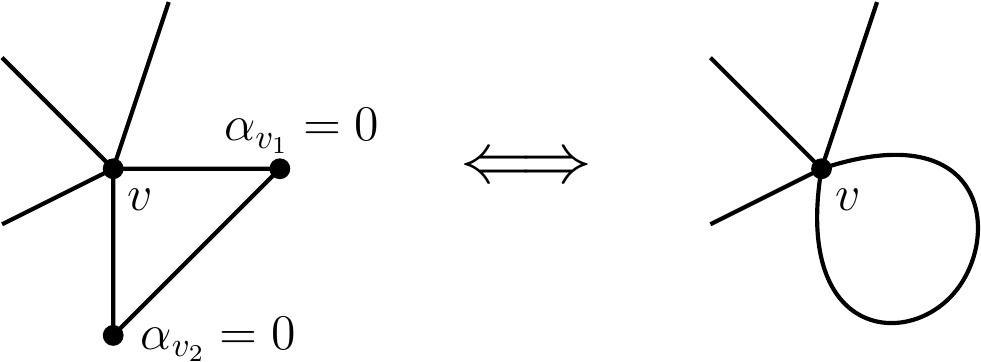}
  \caption{If the vertex conditions $\alpha_{v_1}=\alpha_{v_2}=0$ with
    $\deg(v_1)=\deg(v_2)=2$, the loop is equivalent to a looping edge.}
  \label{fig:loop}
\end{figure}

\begin{example}
  \label{ex:loop_graph}
  Let $L$ be a graph consisting of one looping edge with no vertices.
  We will call such a graph a \emph{circle}.
  By Remark~\ref{rem:vertex_deg2} it can be equivalently represented
  as a cycle graph (a number of vertices connected into a closed
  chain) with all vertices having $\alpha_v=0$.  It is easy to see
  that the spectrum of the graph is
  \begin{equation*}
    0,\ \left(\frac{2\pi}\ell\right)^2,\
    \left(\frac{2\pi}\ell\right)^2,\ 
    \left(\frac{4\pi}\ell\right)^2,\ 
    \left(\frac{4\pi}\ell\right)^2,
    \left(\frac{6\pi}\ell\right)^2, \ldots 
  \end{equation*}
  where $\ell$ is the length of the looping edge.  We note that the
  double degeneracies in the spectrum \emph{cannot} be resolved by
  changing the edge length.

  The eigenfunctions can be represented as
  \begin{equation}
    \label{eq:efun_loop}
    C_1 \cos(\sqrt{\lambda} x) + C_2 \sin(\sqrt{\lambda} x) = A
    \sin(\sqrt{\lambda} x + \theta),
  \end{equation}
  with constants $C_1$ and $C_2$ (or $A$ and $\theta$) arbitrary.
  Here $\lambda$ is the eigenvalue and the origin $x=0$ can be put in
  an arbitrary location on the graph.  It is important to note that for
  any eigenvalue $\lambda \neq 0$ and any point on the graph, there is
  an eigenfunction which vanishes at that point.
\end{example}

\begin{example}
  \label{ex:graph_with_loop}
  Consider a graph $\Gamma$ with a looping edge $L$, see
  Fig.~\ref{fig:loop}.  We assume that there are no other
  (non-trivial) vertices on the loop.  The condition at the
  attachment point $v$ is of $\delta$-type with arbitrary $\alpha_v$.

  If $\ell$ is the length of the loop, then $\left(2\pi n / \ell
  \right)^2$ is an eigenvalue of $\Gamma$ for any integer $n>0$.  We
  demonstrate this by constructing an eigenfunction of $\Gamma$.  On
  the loop we take the function $f$ to be equal to the eigenfunction
  of the corresponding circle, equation~(\ref{eq:efun_loop}),
  chosen to vanish at the attachment point $v$.  The function $f$ is
  extended to the rest of the graph $\Gamma$ by setting it to 0
  identically.  This obviously makes $f$ continuous and, since $f$ is
  an eigenfunction with respect to the loop,
  \begin{equation}
    \label{eq:eig_wrt_loop}
    \sum_{e\in E_v(\Gamma)}\dfrac{df}{dx_e}(v) 
    = \sum_{e\in E_v(L)}\dfrac{df}{dx_e}(v) 
    = 0 = \alpha f(v),
  \end{equation}
  \emph{for any} $\alpha$; the second summation is performed only over
  the edge-ends that belong to the loop L.  We thus have an
  eigenfunction of $\Gamma$ which is supported exclusively on the loop
  $L$; in particular it is zero on all vertices of $\Gamma$.
  Moreover, such an eigenfunction \emph{cannot} be destroyed by
  changing the lengths of graph $\Gamma$.  We also note, that the
  eigenfunction that is supported exclusively on a given loop is
  unique (for a given value of $\lambda$).  This can be easily seen as
  the eigenfunction satisfies the Dirichlet problem on the looping
  edge.
\end{example}

It turns out that having no other vertices on the loop is an
essential feature of Example~\ref{ex:graph_with_loop}.

\begin{lemma}
  \label{lem:support_is_not_a_circle}
  Let $\Gamma$ be a graph with $\delta$-type conditions at vertices.
  Suppose $L$ is loop in $\Gamma$ which has at least one vertex with
  $\alpha_v \neq 0$ on it, other than the attachment vertex.  Then
  there is a small modification of edge lengths of $\Gamma$, after
  which $\Gamma$ has no eigenfunctions $f$ supported exclusively on
  the loop $L$.
\end{lemma}

This lemma, proved in Section~\ref{sec:proof_lem_loop}, motivates the
following definition.

\begin{definition}
  A \emph{pure loop} is a loop with no vertices $v_j$ having
  $\alpha_{v_j}\neq 0$, other than, possibly, the attachment point
  $v$.  In fact, in what follows, by a ``loop'' we will always mean a pure
  loop, unless explicitly stated otherwise.  As mentioned already, a
  graph consisting of one pure loop is called a \emph{circle}; a graph
  consisting of an impure loop will be called an \emph{impure loop
    graph}.
\end{definition}

Now we are able to formulate our main result.

\begin{thm}
  \label{thm:main}
  Let $\Gamma$ be a connected graph with $\delta$-type conditions at
  vertices.  If $\Gamma$ is not equivalent to a circle, then,
  after a small modification of edge lengths, the new graph
  $\tilde{\Gamma}$ will satisfy the following genericity conditions
  \begin{enumerate}[(i)]
  \item \label{item:main_simple} $\sigma(\tilde{\Gamma})$ is simple, and
  \item \label{item:eigenfunctions} for each eigenfunction $f$ of $\tilde{\Gamma}$,
    \begin{enumerate}[(a)]
    \item \label{subitem:nonzero} either $f(v)\neq 0$ for each vertex $v$, or
    \item \label{subitem:supp_loop} $\supp f=L$ for only one loop $L$ of $\tilde{\Gamma}$.
    \end{enumerate}
  \end{enumerate}
  More precisely, in the space of all possible edge lengths, the set
  on which the above conditions are satisfied is residual (comeagre).
\end{thm}

\begin{remark}
  \label{rem:residual}
  A \emph{residual} or \emph{comeagre} set is a countable intersection
  of sets with dense interiors.  Informally, a residual set is
  ``large''.  In particular, since all spaces we will be dealing with
  (namely, the space of all possible lengths or the space of all
  points on a graph) are complete metric spaces, by Baire Category
  Theorem a residual set is dense.  In particular, it implies that the
  adjustments in the edge lengths, promised in the Theorem, can be
  chosen arbitrarily small.

  In general, the converse is not true, a dense set is not necessarily
  residual (for example, the set of rationals is dense yet meagre).
  But in our case, it is in fact enough to establish the Theorem for a
  dense set of lengths.  Indeed, once that has been established, one
  can argue as follows.  For any $N$, the set of lengths on which the
  conclusions of the Theorem are true for all $n < N$ is still dense
  but it is also open: both eigenvalues and eigenfunctions are
  continuous under change of lengths
  \cite{BerKuc_incol12,BerKuc_graphs} and a small perturbation will
  preserve strict inequalities between the eigenvalues, and between values
  of $f(v)$ and 0.  Taking the intersection of open dense sets over
  countably many $N$ we obtain a residual set.
\end{remark}

\section{Preliminary observations}

Before we start the proof of the main theorem, we make some observations.

\begin{remark}     
  \label{rem:lambda0}
  If $\lambda=0$ and a graph has no Robin-type vertices (i.e.\ only
  $\alpha_v=0$ or $\alpha_v=\infty$ are allowed), it can be easily
  shown (see, e.g. \cite[Thm 1]{Kur_ark08}) that the corresponding
  eigenfunction $f$ of $\Gamma$ is a constant on every edge.
  Therefore the multiplicity of 0 in the spectrum is at most the
  number of connected components in the graph, which is one in our
  case.  If there is a vertex with Dirichlet condition $\alpha_v=0$,
  the value 0
  is not an eigenvalue.  Hence, Theorem~\ref{thm:main} holds for
  $\lambda=0$ on graphs with no Robin-type conditions.

  The same proof does not apply to graphs with some $\alpha_v \neq 0,
  \infty$, which introduces a layer of complication into the proof of
  our main result.
\end{remark} 

The proof of Theorem~\ref{thm:main} is built around modifications made
to the structure of a graph.  The following theorem describes one of
the modification we find useful and its effect on the spectrum.  We
denote by $\Gamma_{\alpha}$ a compact quantum graph with a
distinguished vertex $v$.  Arbitrary self-adjoint conditions are fixed
at all vertices other than $v$, while $v$ is endowed with the
$\delta$-type condition with coefficient $\alpha$.

\begin{thm}[Berkolaiko--Kuchment \cite{BerKuc_incol12} and
  {\cite[Thm 3.1.8]{BerKuc_graphs}}]
  \label{thm:interlacing_delta}
  Let $\Gamma_{\alpha'}$ be the graph obtained from the graph
  $\Gamma_{\alpha}$ by changing the coefficient of the condition at
  vertex $v$ from $\alpha$ to $\alpha'$. If
  $-\infty \leq \alpha < \alpha' \leq \infty$ (where $\alpha'=\infty$
  corresponds to the Dirichlet condition), then
  \begin{equation}
   \lambda_{n-1}(\Gamma_{\alpha'}) \leq \lambda_n(\Gamma_{\alpha})
   \leq \lambda_{n}(\Gamma_{\alpha'}).
  \end{equation}
  If the eigenvalue $\lambda_n(\Gamma_{\alpha})$ is simple and its
  eigenfunction $f$ is such that either $f(v)$ or $\sum f'(v)$ is
  non-zero, then the inequalities can be made strict,
  \begin{equation}
    \label{eq:interlacing_strict}
    \lambda_{n-1}(\Gamma_{\alpha'}) < \lambda_n(\Gamma_{\alpha})
    < \lambda_{n}(\Gamma_{\alpha'}).
  \end{equation}

  If $\alpha' < \alpha$, the inequalities are adjusted accordingly,
  \begin{equation}
    \lambda_{n}(\Gamma_{\alpha'}) < \lambda_n(\Gamma_{\alpha})
    < \lambda_{n+1}(\Gamma_{\alpha'}).
  \end{equation}
\end{thm}

Another modification we will use is splitting a vertex into two.  A
vertex $v$ is replaced by two vertices, $v_1$ and $v_2$ which, among
them, split the set of edges originally incident to the vertex $v$:
$E_v$ is a disjoint union of $E_{v_1}$ and $E_{v_2}$.  The
$\delta$-type constant at the new vertices is chosen such that
$\alpha_{v_1}+\alpha_{v_2} = \alpha_v$ (usually $\alpha_{v_1}$ will be
taken to be 0).  The key observation here is that if an eigenfunction
$f$ satisfies the sum of derivatives condition
\eqref{eq:cond_deltatype} with respect to the subset $E_{v_1} \subset
E_v$, it will automatically satisfy it with respect to $E_{v_2}$ and
will therefore be an eigenfunction of the modified graph with the
additional property that $f(v_1) = f(v_2)$.  To arrive to a
contradiction, we will need to show that the latter is unlikely to
happen.

\begin{lemma}
  \label{lem:nonzero}
  Let $x$ be a point on the edge $e$.  Then, in any neighborhood of
  $x$ there is a residual set of $y$ such that for all eigenfunctions $f_n$
  \emph{either} $f_n(y) \neq 0$ \emph{or} $f_n\equiv0$ on the edge $e$.  

  Similarly, given a sequence of values $\{\phi_n\}$, all of them
  non-zero, there is a residual set of $y$ such that for all
  normalized eigenfunctions $f_n$ with $\lambda_n > 0$
  we have $f_n(y) \neq \phi_n$.
\end{lemma}

\begin{proof}
  Fix a neighborhood of $x$.  Any eigenfunction $f$ that is not
  identically zero on $e$ has only finitely many zeros in the
  neighborhood: otherwise there is an accumulation point for zeros at
  which $f=f'=0$ and therefore $f\equiv0$.  The union of the zero
  points over all possible $n$ is a countable set; its complement is residual.

  The second part of the lemma is proved analogously, only one need
  not worry about $f_n$ vanishing identically on the edge.
\end{proof}

There is one example of a graph for which our modifications fail.  We
consider this example separately.

\begin{example}
  \label{ex:figure8}
  Consider the ``figure of 8'' graph with Neumann--Kirchhoff condition
  at the vertex of degree 4, shown in Fig.~\ref{fig:figure8}.  The
  lengths of the two loops will be denoted by $\ell_1$ and $\ell_2$.
  From Example~\ref{ex:graph_with_loop} we know that there two classes
  of ``loop eigenfunctions'' with eigenvalues $\left\{ (2\pi
    n/\ell_1)^2 \right\}_{n\geq1}$ and $\left\{ (2\pi
    n/\ell_2)^2\right\}_{n\geq1}$.  One can also construct an
  eigenfunction by starting with an eigenfunction on a circle of
  length $\ell_1+\ell_2$ and pinching the circle at the points of
  distance $\ell_1/2$ from a maximum (the eigenfunction takes equal
  values at these points) to create an eigenfunction of the figure of
  8 graph.  Therefore $\left\{ \left({2\pi n}/({\ell_1 +
        \ell_2})\right)^2 \right\}_{n\geq 0}$ are also eigenvalues.
  Symmetry considerations (there is a basis of eigenfunctions which
  are either even or odd with respect to the flip of a loop) show that
  the above choices exhaust the set of eigenfunctions.

  \begin{figure}
    \centering
    \includegraphics{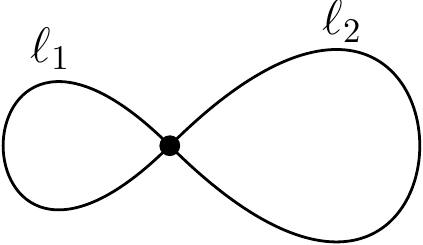}
    \caption{A "figure of 8" graph.  The spectrum is simple if and
      only if the lengths $\ell_1$ and $\ell_2$ are rationally
      independent.}
    \label{fig:figure8}
  \end{figure}

  Now, if the numbers $\ell_1$ and $\ell_2$ are rationally
  independent, the three sets above are mutually disjoint and the
  spectrum of the figure of 8 graph is simple.  Conclusion (ii) of
  Theorem~\ref{thm:main} also holds with the third class of
  eigenfunctions.
\end{example}

\section{Proofs of the main results}

\subsection{Proof of Lemma \ref{lem:support_is_not_a_circle}}
\label{sec:proof_lem_loop}

We begin by establishing the following auxiliary result.

\begin{lemma}\label{lem:simplicity_of_a_loop}
  Let $\Gamma$ be an impure loop graph, i.e. a graph consisting of one
  loop with at least one vertex with coefficient $\alpha_v \neq
  0$.  Then, for a residual set of edge lengths, the eigenvalues of
  the graph are simple.
\end{lemma}

\begin{proof}
  We prove the lemma by induction on the number of vertices with
  $\alpha_v\neq 0$.

  Let $\Gamma_1$ be a loop with one vertex $v$ and $\ell$ be the
  length of the loop.  Parametrize the loop with a coordinate $x$ such
  that $x=0$ corresponds to $v$, $x>0$ in the clockwise direction and
  $x<0$ in the anticlockwise direction, with $\ell/2$ ad $-\ell/2$
  corresponding to the same point.  Since $\Gamma_1$ has reflection
  symmetry, every eigenfunction is either odd or even, see
  Fig.~\ref{fig:even_odd_impure}.

  \begin{figure}
    \center
    {
    \includegraphics[width=5cm]{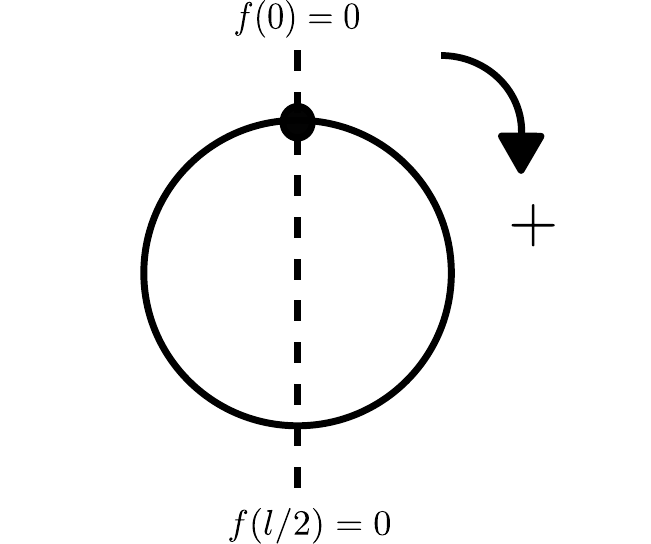}
    \hspace{1cm}
    \includegraphics[width=5cm]{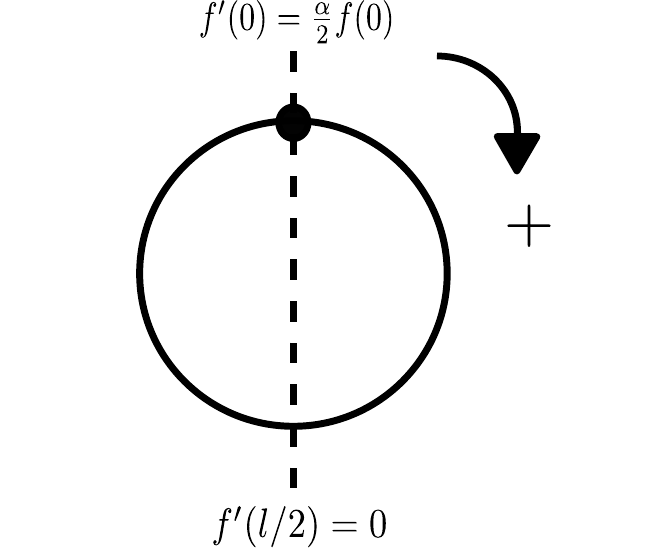}
    }
    \caption{A loop with one vertex and the structure of the odd (left)
      and even (right) eigenfunctions.}
    \label{fig:even_odd_impure}
  \end{figure}

  If $f$ is odd, it satisfies $f(-x)=-f(x)$ for each point $x$.  In
  particular, $f(0)=0$ and, by continuity, $f(\ell/2)=f(-\ell/2)=-f(\ell/2) =
  0$.  Solving the equation $Hf=\lambda f$, we have $f(x)=\sin(\sqrt{\lambda}x)$ and
  \begin{equation}\label{eqn:circle_with_one_vertex_case1}
    \sqrt{\lambda} = 2k\pi /\ell, \qquad k\in \mathbb{N}.
  \end{equation}

  If $f$ is even, i.e. $f(-x)=f(x)$ for every point $x$.
  Since $f'(\ell/2)=-f'(-\ell/2)$ by symmetry and $f'(\ell/2)-f'(-\ell/2)=0$ by
  NK vertex condition at $\ell/2$, we have $f'(\ell/2)=0$.  At $x=0$,
  \begin{equation*}
    \sum_{e\in E_0} \dfrac{df}{dx_e}(0)=\alpha_0 f(0),
  \end{equation*}
  i.e. $2 f_+'(0)=\alpha_0 f(0)$, where $f'_+$ denotes the one-sided
  derivative taken at $0$ in the positive direction.  Solving the
  equation $Hf=\lambda f$, we have
  \begin{equation}\label{eqn:circle_with_one_vertex_case2}
    2\sqrt{\lambda}\sin(\sqrt{\lambda}\ell/2) = \alpha_0\cos(\sqrt{\lambda}\ell/2).
  \end{equation}
  The roots of equation (\ref{eqn:circle_with_one_vertex_case2})
  cannot coincide with (\ref{eqn:circle_with_one_vertex_case1}): the
  substitution of $\sqrt{\lambda} = 2k\pi /l$ into
  (\ref{eqn:circle_with_one_vertex_case2}) results in $0 = \pm
  \alpha_0$, which contradicts our assumptions.  Hence we proved the
  base case for the induction.

  Suppose the statement is true for any impure loop graph with $n$
  Robin-type vertex conditions.  Consider $\Gamma$, an impure loop graph
  with $n+1$ nonzero vertex conditions.  Pick any vertex $v$ and
  change $\alpha_v$ to 0; using inductive hypothesis,
  adjust the edge lengths to obtain a graph $\Gamma'$ with simple
  spectrum.  By Lemma~\ref{lem:nonzero} it is now possible to pick a
  point near the former position of the vertex $v$, where none of the
  eigenfunctions are zero.  Note that the eigenfunctions cannot vanish
  on an open subset of the graph, since the unique continuation holds
  for the impure loop graph.  Now we change the vertex condition at
  the new $v$ back to $\alpha_v$ and use the strict inequalities in
  Theorem~\ref{thm:interlacing_delta} to conclude that the spectrum is
  still simple.

  We have established that there is an arbitrarily small perturbation
  of edge lengths which will make the spectrum simple.  Therefore, the
  set of admissible lengths is dense and, by the argument in
  Remark~\ref{rem:residual}, we deduce that the set of admissible
  lengths is residual.
\end{proof}

\begin{remark}
  The proof of Lemma~\ref{lem:simplicity_of_a_loop} allows for a
  slightly stronger statement: the length modifications actually
  preserve the total length of the graph.
\end{remark}

We are now ready to prove Lemma~\ref{lem:support_is_not_a_circle}.

\begin{proof}[Proof of lemma \ref{lem:support_is_not_a_circle}]
  Split the loop $L$ at the attachment point $v$ from $\Gamma$.
  Assign NK vertex condition to the former attachment point on
  the loop, which we now call $v_1$, see
  Fig.~\ref{fig:proof_impure_loop}, and keep other vertex conditions unchanged.

 \begin{figure}
    \centering
    \includegraphics[scale=0.7]{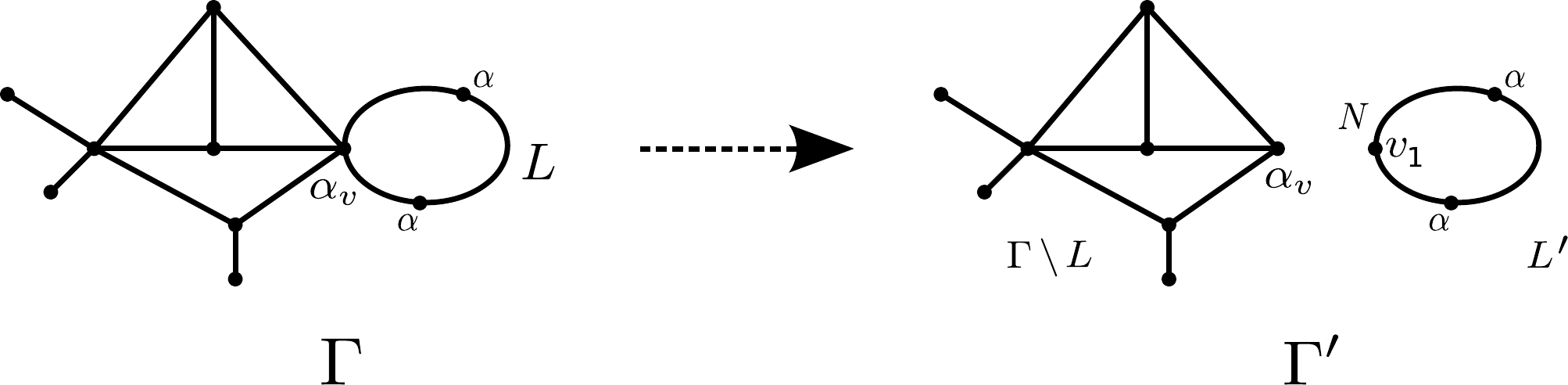}
    \caption{Modifications to graph $\Gamma$ for the proof of
      Lemma~\ref{lem:support_is_not_a_circle}.}
    \label{fig:proof_impure_loop}
  \end{figure}

  Apply Lemma \ref{lem:simplicity_of_a_loop} to $L$ so that
  $\sigma(L')$ is simple for the changed length loop $L'$.
  Furthermore, by Lemma~\ref{lem:nonzero}, we can pick a point $v_1'$
  arbitrarily close to the former attachment point so that each
  eigenfunction $f$ of $L'$ is nonzero at $v_1'$.  Attach $L'$ back to
  $\Gamma\setminus L$ at $v_1'$. Then the new graph $\Gamma'$
  satisfies the same vertex conditions everywhere as $\Gamma$.

  If there exists an eigenfunction $g$ of $\Gamma'$ with $\supp g=L'$,
  then necessarily $g(v_0)=0$ and $g$ is an eigenfunction of the loop
  $L'$ (see equation~(\ref{eq:eig_wrt_loop})), which is a
  contradiction.
\end{proof}

Now, we are ready to prove the main theorem.

\begin{proof}[Proof of theorem \ref{thm:main}]
  We will prove the result by an induction on the number of edges of
  the graph.  We will also employ a sub-induction on the number of
  vertices with Robin-type conditions.

  If $\Gamma$ consists of one edge which is not a loop, the statement
  holds by the classical Sturm-Liouville theory.  The case of a loop
  with no non-trivial $\delta$-type vertices is specifically excluded
  by the assumptions of the Theorem.  The case of a loop with one
  vertex $v$ with a non-zero $\delta$-type condition is covered by
  Lemma~\ref{lem:simplicity_of_a_loop} (part~(\ref{subitem:supp_loop})
  of the Theorem is true automatically).  A loop with more than one
  non-zero condition is already a graph with at least two edges.

  The plan for the inductive step is as follows.  First we establish
  part~(\ref{item:main_simple}) for a graph $\Gamma$ if \emph{both
    parts} of the Theorem hold for every graph with a smaller number
  of edges or with the same number of edges and a smaller number of
  vertices with Robin-type conditions.  Then we will establish
  part~(\ref{item:eigenfunctions}) assuming, in addition to the above,
  that the spectrum of $\Gamma$ is simple.

  Consider $\Gamma$, a connected graph with $n$ edges, satisfying
  $\delta$-type conditions with coefficients $\alpha_v$ for each
  vertex $v$.  For the proof of part~(\ref{item:main_simple}) we
  consider three cases.

  \emph{Part~(\ref{item:main_simple}), case 1}.  $\Gamma$ has no loops
  or cycles, i.e. it is a tree.

  \begin{figure}
    \center
    \includegraphics[scale=0.7]{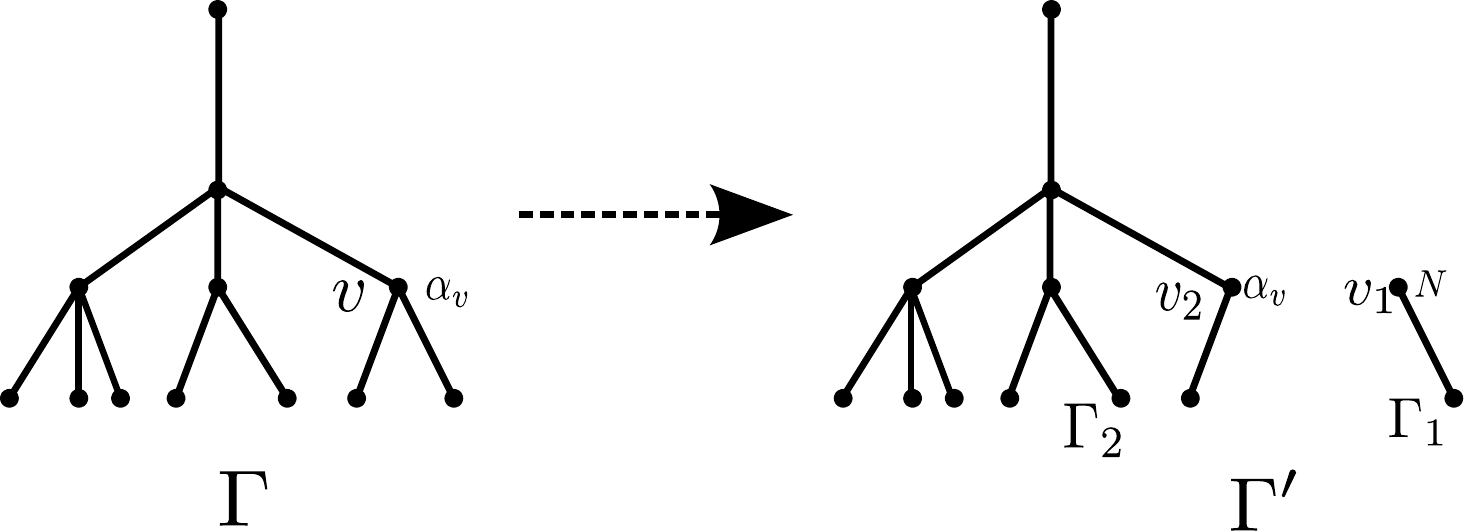}
    \caption{Part (i), case 1: splitting away an edge from a tree.}
    \label{fig:tree_case}
  \end{figure}

  Choose an edge $e$ leading to a leaf (a vertex of degree 1) of the
  tree and split it from the tree.  The new one-edge graph we denote
  by $\Gamma_1$ while the rest of the tree is denoted by $\Gamma_2$.
  The attachment point $v$ of that edge is split into two vertices,
  $v_1\in\Gamma_1$ and $v_2\in\Gamma_2$, see Fig.~\ref{fig:tree_case}.
  We assign NK condition to vertex $v_1$.  The vertex $v_2$
  inherits the $\delta$-type condition with the constant $\alpha_v$
  (which may also be 0), while all other vertices keep their previous
  conditions.

  Adjust the edge lengths of
  $\Gamma_1$ and $\Gamma_2$ so that
  \begin{enumerate}
  \item the graph $\Gamma_2$ satisfies (i) and (ii), and
  \item $\sigma(\Gamma_1)\cap\sigma(\Gamma_2)=\{0\}$.  
  \end{enumerate}

  Note that $\sigma(\Gamma_1)$ is a set of strictly decreasing
  functions of the edge length.  Once condition (1) has been satisfied
  (which is possible on a residual set by the inductive hypothesis)
  and the discrete set $\sigma(\Gamma_2)$ has been fixed, the
  intersection $\sigma(\Gamma_1)\cap\sigma(\Gamma_2)$ is non-empty on
  a countable set of lengths of $\Gamma_1$, and its complement is
  dense.  Altogether, conditions (1) and (2) above are satisfied on a
  dense set of edge lengths of the graph $\Gamma$.  We remark that the
  graph $\Gamma_2$ satisfies properties (i) and (ii) of the Theorem
  automatically.

  Glue $v_1$ and $v_2$ back together and call the resulting vertex
  $\tilde{v}$.  The new graph $\tilde{\Gamma}$ has the same vertex
  conditions as $\Gamma$.  We claim that the spectrum of $\Gamma$ is
  simple.  Assume the contrary, $\lambda\in\sigma(\tilde{\Gamma})$ is
  multiple with corresponding eigenfunctions $f_i$.  Then we can find
  a non-zero linear combination $f=\sum a_i f_i$, which is still an
  eigenfunction of $\tilde{\Gamma}$, such that it is also an
  eigenfunction \emph{with respect to the graph $\Gamma_1$}, i.e.\
  $f'(\tilde{v})=0$ along the edge $e$.  Since
  \begin{equation*}
    \alpha_v f(\tilde{v}) = \sum_{e'\in E_{\tilde{v}}} \dfrac{df}{dx_{e'}}(\tilde{v}) 
    = \sum_{e'\in E_{\tilde{v}} \setminus \{e\}} \dfrac{df}{dx_{e'}}(\tilde{v}),    
  \end{equation*}
  the function $f$ is also an eigenfunction with respect to the graph
  $\Gamma_2$. 

  If $f$ is non-zero on both $\Gamma_1$ and $\Gamma_2$, condition (2)
  is violated. If $f$ is zero on one of them, it is zero on
  $\tilde{v}$ which violates condition (1). Therefore, the spectrum of
  $\tilde{\Gamma}$ is simple.

  \emph{Part (i), case $2$.}  $\Gamma$ contains at least one loop.

  \begin{figure}
    \center
    \includegraphics[scale=0.7]{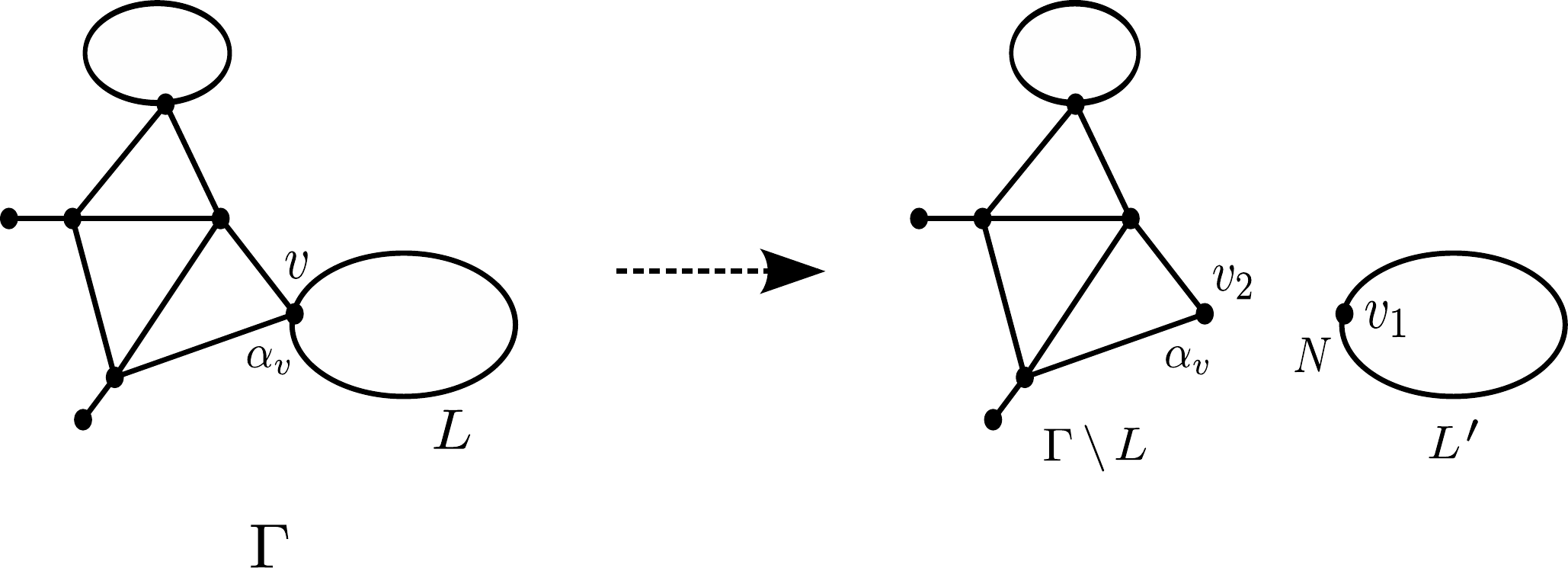}
    \caption{Part (i), case 2: splitting away a loop.}
    \label{fig:loop_case}
  \end{figure}

  For each loop $L$ of $\Gamma$, we repeat the steps of the previous case, namely,
  \begin{enumerate} 
  \item split $L$ away from the rest of $\Gamma$ at the attachment
    point $v$, see Fig.~\ref{fig:loop_case}. For $\Gamma\setminus L$, keep
    $\alpha_v$ vertex condition for the attachment point $v_2$ and
    keep all other vertex conditions unchanged.  Adjust the lengths of
    $\Gamma\setminus L$ so that
    \begin{equation*}
      \Gamma\setminus L \mbox{ satisfies (i) and (ii)},
    \end{equation*}
  \item \label{item:disjoint_spec} assign NK vertex condition to
    the attachment point $v_1$ on $L$; adjust $L$ so that
    \begin{equation*}
      \sigma(L)\cap\sigma(\Gamma\setminus L)=\{0\},
    \end{equation*}
    and glue $L$ back to $\Gamma$.
  \end{enumerate}
  For each loop the above condition (1) is satisfied on a residual set
  of lengths of the graph $\Gamma\setminus L$ and for each $\Gamma
  \setminus L$, condition (2) is satisfied on a dense set of lengths
  of $L$.  Altogether they are satisfied on a dense set of lengths of
  $\Gamma$, but repeating the argument of Remark~\ref{rem:residual} we
  conclude that the set is actually residual.

  For all loops $L$ \emph{simultaneously}, the conditions are
  satisfied on an intersection of residual sets which is also
  residual.

  Let us first consider $\lambda$ which is an eigenvalue of some loop
  $L$.  By Example~\ref{ex:graph_with_loop}, $\lambda = (2\pi n/\ell)^2 \in
  \sigma(\Gamma)$, where $\ell$ is the length of the loop graph $L$.
  As explained in Example~\ref{ex:graph_with_loop}, the eigenfunction
  supported on the loop is unique.

  We want to show that the above $\lambda$ is simple in the spectrum
  $\sigma(\Gamma)$. Assume the contrary, there is at least one
  eigenfunction $f$ which is not identically zero on $\Gamma\setminus
  L$.  Transform $f$ by flipping the loop; this is still an
  eigenfunction which we will denote by $\tilde{f}$.  The function $g
  = (f+\tilde{f})/2$ has the following properties: it is an
  eigenfunction of $\Gamma$, not identically zero on $\Gamma\setminus
  L$, and it is even with respect to flipping the loop.  The latter
  implies that its derivative at the midpoint of the loop is zero
  \begin{equation*}
    g_L'(\ell/2) = 0.
  \end{equation*}
  Parametrizing the loop as $[0,\ell]$, we know that $g$ on the loop
  takes the form $g_L = A\cos(2\pi n (x-\ell/2) / \ell)$ which by
  direct computation implies that
  \begin{equation*}
    g_L'(0) = g_L'(\ell) = 0.
  \end{equation*}
  Therefore, the function $g$ satisfies the $\delta$-type conditions
  at the attachment point $v$ also with respect to the graph
  $\Gamma\setminus L$.  Thus we get $\lambda \in \sigma(\Gamma
  \setminus L)$ in contradiction with condition
  (\ref{item:disjoint_spec}) above.

  Now, if $\lambda$ is not an eigenvalue of any loop of $\Gamma$, and
  is multiple, we can find $f=\sum_{i=1}^2 a_i f_i$ such that $f$
  satisfies NK vertex condition with respect to some loop $L$.
  It must be identically zero on $L$ (otherwise $\lambda$ is an
  eigenvalue of $L$), therefore at the attachment point $f(v)=0$.
  Since $f$ is an eigenfunction on the graph $\Gamma\setminus L$,
  which satisfies (ii), it must be supported on some other loop $L'$,
  resulting in a contradiction.

  We remark that there is one case for which the graph
  $\Gamma\setminus L$ is not covered by the inductive hypothesis
  because it is a circle: the ``figure of 8'' graph considered
  explicitly in Example~\ref{ex:figure8}.

  \emph{Part (i), case $3$.} $\Gamma$ contains at least one cycle and no loops.
  
  Now we also employ the induction on the number of Robin-type
  (i.e. $\alpha_v \neq 0, \infty$) vertices.  First consider the base
  case of no Robin-type vertices.   For this case we can assume that
  $\lambda \neq 0$, as the case $\lambda=0$ is already covered by
  Remark~\ref{rem:lambda0}. 

  \begin{figure}
    \center
    \includegraphics[scale=0.7]{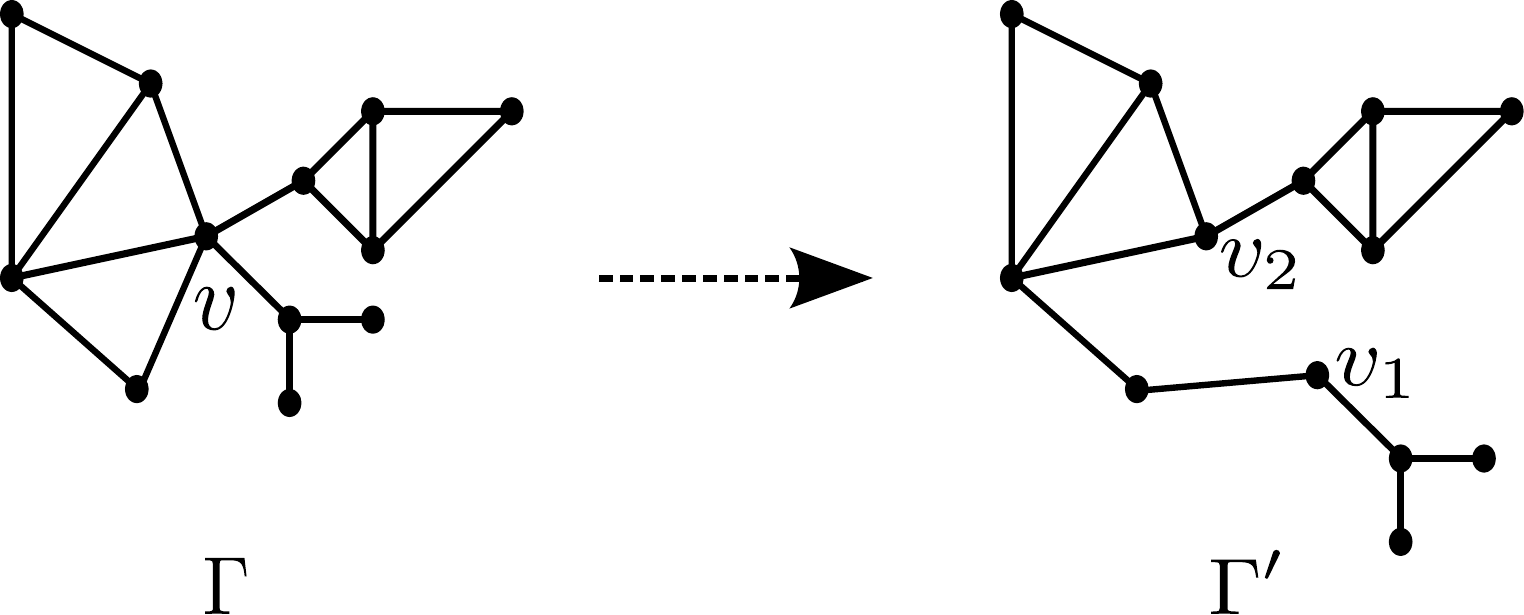}
    \caption{Part (i), case $3$: Splitting a graph with a cycle; each of
      the new vertices is not on a loop.}
    \label{fig:cyclecase}
  \end{figure}
  
  Pick a vertex $v$ on the cycle such that $\deg(v)\geq 3$.  Split
  $\Gamma$ at $v$ so that $\deg(v_1)=2$, $\deg(v_2)\geq 1$, and the
  graph is still connected (this is possible precisely because $v$ is
  on a cycle).  For the new graph $\Gamma'$, assign NK vertex
  condition to $v_1$ and $v_2$, and keep other vertex conditions
  unchanged, see Fig.~\ref{fig:cyclecase}. The vertex $v_1$ is now
  trivial, hence $\Gamma'$ has effectively one edge less than
  $\Gamma$.  We can thus use induction and adjust $\Gamma'$ to satisfy
  conditions (i) and (ii).
  
  \begin{figure}
    \center
    \includegraphics[scale=0.7]{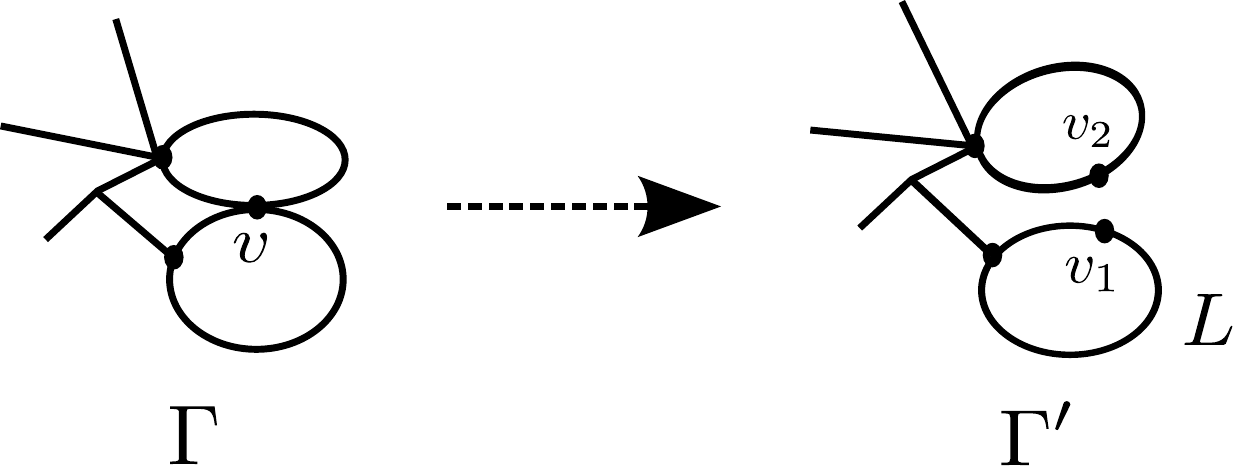}
    \caption{Part (i), case $3$: Splitting a graph with a cycle; both
      new vertices are on loops.}
    \label{fig:cyclecase_twoloops}
  \end{figure}
  
  First, we show that we can always find a new vertex $v_1'$ near
  $v_1$ so that $f(v_1')\neq f(v_2)$ for each non-constant
  eigenfunction $f$ of $\Gamma'$.  Because our modification of
  $\Gamma$ may have created loops, we need to consider three
  possibilities:
  \begin{enumerate}
  \item neither $v_1$ nor $v_2$ is on a loop. Then $\Gamma'$ has no
    loops. Since $\Gamma'$ satisfies (i) and (ii), $f(v_2)\neq 0$.  By
    Lemma~\ref{lem:nonzero}, there exists $v_1'$ so that $f(v_1')\neq
    f(v_2)$.
  \item only one of $v_1$ and $v_2$ is on a loop. After relabeling, we
    may assume that $v_1$ is on a loop we denote $L$ and $v_2$ is not on
    a loop.  Since $\Gamma'$ satisfies (i) and (ii), for each
    eigenfunction, is not identically 0 in the neighborhood of $v_1$ and
    we can again apply Lemma~\ref{lem:nonzero}, whether $f(v_2)$ is zero
    or not.
  \item both of $v_1$ and $v_2$ are on loops, see
    Fig.~\ref{fig:cyclecase_twoloops}.  Then both $v_1$ and $v_2$ may
    be adjusted.  Since the eigenfunction cannot vanish identically
    around both $v_1$ and $v_2$ at the same time, we can again use
    Lemma~\ref{lem:nonzero} to make adjustments until $f(v_1') \neq
    f(v_2')$.
  \end{enumerate} 

  Now glue $v_1'$ and $v_2$ (or $v_2'$, if appropriate) together and
  call the vertex $\tilde{v}$.  Note that the new graph $\tilde{\Gamma}$
  has the same vertex conditions as $\Gamma$.  Assume that
  $\lambda\in\sigma(\tilde{\Gamma})$ is multiple with eigenfunctions
  $f_i$.  Similarly to before, we can find $f=\sum a_i f_i$, still an
  eigenfunction of $\tilde{\Gamma}$, such that it satisfies
  Neumann--Kirchhoff condition with respect to edges that were connected to $v_1'$,
  \begin{equation*}
    \sum_{e\in E_{v_1'}}\dfrac{df}{dx_e}(v_1')=0.
  \end{equation*}
  Then $f$ also satisfies the NK vertex condition with respect to the
  edges connected to $v_2$.  Therefore, $f$ is an eigenfunction of
  $\Gamma'$, which means $f(v_1')\neq f(v_2)$, contradicting the fact
  that $f$ is continuous at $\tilde{v}$ for $\tilde{\Gamma}$. Hence
  $\lambda$ is simple.

  Finally we consider the case of graphs without cycles but with some
  Robin-type vertices.  Choosing an arbitrary vertex $v$ with $\alpha_v
  \neq 0$, we set $\alpha_v = 0$ and obtain a graph $\Gamma'$ with
  less Robin-type vertices for which, as we assumed, the Theorem
  holds.  We adjust the edge lengths to obtain a graph
  $\tilde{\Gamma}'$ both properties (i) and
  (ii), in particular we achieve that $f(v) \neq 0$ for every
  eigenfunction of $f$ of the graph (using, if necessary,
  Lemma~\ref{lem:nonzero}). Changing the parameter $\alpha_v$ back to
  its original non-zero value we use
  Theorem~\ref{thm:interlacing_delta} in its strict form
  (equation~\eqref{eq:interlacing_strict}) to conclude that the
  spectrum of the graph $\tilde{\Gamma}$ is simple.

  \emph{Proof of part (ii)}. We will show the statement on a single
  vertex basis, that is we fix a vertex $v$ and show that each
  eigenfunction is either non-zero at $v$ or is supported on a loop.
  This is achieved by small modifications of the graph and holds on a
  residual set of lengths.  As a result it will hold at every $v$ on an
  intersection of residual sets, which is also residual.
  
  First, we assume that the vertex $v$ has $\alpha_v=0$.  We may
  assume that the spectrum of $\Gamma$ is simple.  Also, after a series
  of modifications we may assume that for each edge $e$ of the graph,
  $\Gamma\setminus e$ satisfies (i) and (ii).  If removing $e$
  disconnects the graph, we assume that each of the two components
  satisfies the assumptions.  In the special case when removing an edge
  creates a new loop, we also ask that the loop states do not vanish at the
  point where the edge was attached (again achievable by a small
  movement of the attachment point).  Each of the conditions can be
  fulfilled on a residual set.

  Now let $(\lambda,f)$ be the $n$-th eigenpair of $\Gamma$ such that
  $f$ is the first eigenfunction that vanishes at the chosen vertex
  $v$, $f(v)=0$.

  \emph{Case 1:} $f\vert_e \equiv 0$ for some edge $e$ incident to $v$.

  Now $f(u)=f(v)=0$ for the end points $u, v$ of the edge $e$ and $f$ is
  also an eigenfunction of $\Gamma\setminus e$.  Since $\Gamma\setminus
  e$ satisfies (i) and (ii), $\supp f=L$ for a loop $L$ in
  $\Gamma\setminus e$.  If this loop is present in $\Gamma$, we have
  nothing further to prove.  If the loop is not present in $\Gamma$,
  then either $u$ or $v$ lies on the loop and we get a contradiction
  with the conditions imposed on $\Gamma\setminus e$ above.
  
  \emph{Case 2:} $f\vert_e \not\equiv 0$ for each edge $e$ incident to
  $v$.
  
  \begin{figure}
    \includegraphics[scale=1]{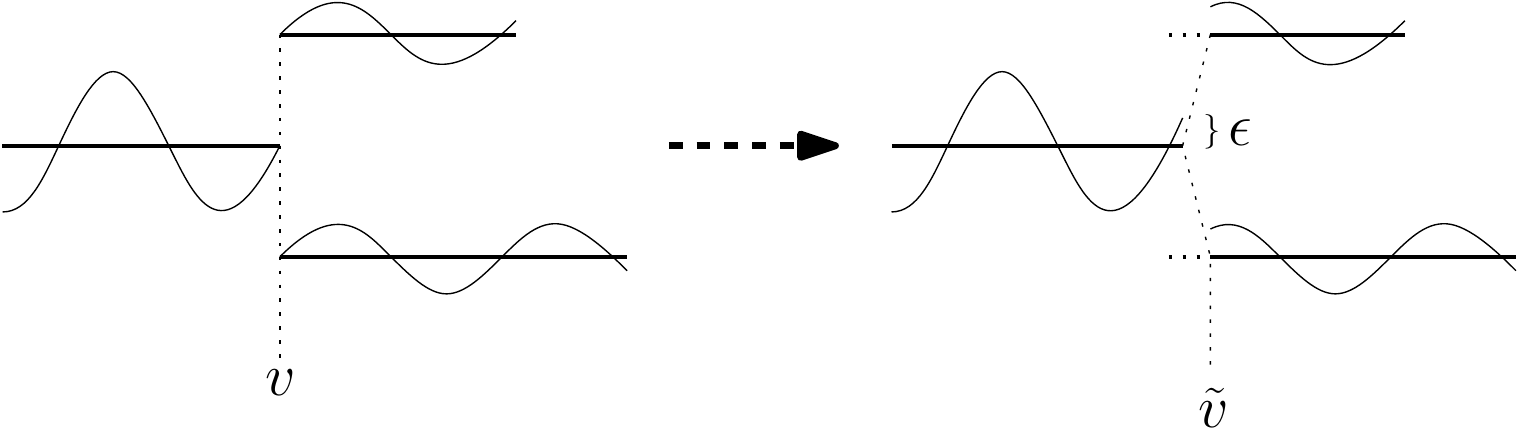}
    \caption{Modifying the edge lengths so that $f_e(v)$ becomes equal to
      $\epsilon$.}
    \label{fig:from0}
  \end{figure}
  
  Parametrizing the edges incident to $v$ so that $x=0$ at $v$, we have
  \begin{equation}
    -\dfrac{d^2}{dx^2} f = \lambda f, \qquad
    \begin{cases}
      \sum_{e\in E_v}\dfrac{df}{dx_e}(0)=0,\\
      f(0) = 0,
    \end{cases}
  \end{equation}
  and, denoting $k = \sqrt{|\lambda|}$, we get
  \begin{equation*}
    f_e(x)=
    \begin{cases}
      A_e\sin(kx), & \lambda >0, \\
      A_ex, & \lambda = 0,\\
      A_e\sinh(kx), & \lambda <0
    \end{cases}
    \qquad\mbox{ and }\qquad 
    \sum_{e\in E_v} A_e=0.    
  \end{equation*}
  Note that each coefficient $A_e\neq0$ since $f\vert_e \not\equiv 0$.
  At $v$, we will shorten the edges with $f_e' = A_e > 0$ and lengthen
  edges with $f_e'<0$ in a way which will be controlled by a (small)
  parameter $\epsilon$.  Namely, we ask that $f_e(\tilde v)=\epsilon$,
  where $\tilde v$ denotes the new position of the vertex $v$ and the
  function $f_e$ is kept as before, see Fig.~\ref{fig:from0}.
  
  In this way, $f$ is still continuous at $\tilde v$ and satisfies
  $\delta$-type vertex condition for some parameter $\alpha'$, which we
  compute as follows.  The new position of the vertex $\tilde v$ on edge $e$
  is determined by the equation $f_{e}(x_e)=\epsilon$, or
  \begin{equation*}
    x_e=
    \begin{cases}
      1/k\arcsin(\epsilon/A_e), & \lambda > 0, \\
      \epsilon / A_e, & \lambda = 0, \\
      1/k\arcsinh(\epsilon/A_e), & \lambda < 0.
    \end{cases}
  \end{equation*}
  We now find the coefficient $\alpha_v'$ from the condition 
  \begin{equation*}
    \sum_{e\in E_v} \dfrac{df}{dx_e}(x_e) = \alpha'_v f(x_e)=\alpha'_v\epsilon,  
  \end{equation*}
  leading to, by Taylor expansion in each of the three cases,
  \begin{equation*}
    \alpha'_v = \frac{1}{\epsilon}\sum_{e\in E_v} f_e'(x_e)
    = O(\epsilon).
  \end{equation*}
  We now consider two families of graphs, continuously depending on the
  parameter $\epsilon$: $\Gamma'$ which has the modified edge lengths and the
  NK condition at the vertex $v$ and $\Gamma''$ which in addition
  to changed lengths has $\alpha_v' = \alpha_v'(\epsilon)$ condition
  computed above.
  
  If the parameter $\epsilon$ is small enough, the eigenfunctions
  below $n$-th are still non-zero at $v$.  Furthermore, the eigenvalues of three
  graphs $\Gamma$, $\Gamma'(\epsilon)$ and $\Gamma''(\epsilon)$ are
  still in correspondence.  More precisely, for any $k$ and small enough
  $\epsilon$, the eigenvalue $\lambda_k(\Gamma_1)$ is closer to
  $\lambda_k(\Gamma_2)$ than to any other eigenvalue of $\Gamma_2$ for
  any two of the above three graphs.
  
  We now claim that the $n$-th eigenfunction of $\Gamma'$ doesn't vanish
  on $v$ for any $\epsilon>0$, provided it is small enough.  Indeed, if
  $f(v)=0$, then it is also an eigenfunction of $\Gamma''$ (condition
  (\ref{eq:cond_deltatype}) will be satisfied for any $\alpha_v$) and
  must have index $n$ due to the correspondence of eigenvalues, but we
  explicitly constructed the $n$-th eigenfunction of $\Gamma''$ above to
  have value $\epsilon$ on the vertex $v$.  This completes the proof of
  case 2.
  
  Finally, we have to consider the vertex $v$ with a Robin-type
  condition: $\alpha_v\neq 0$.  In this case we make edge length
  adjustments to the graph $\Gamma_0$ obtained from $\Gamma$ by setting
  $\alpha_v=0$.  Once condition (ii) is satisfied for $\Gamma_0$ at $v$,
  it is also satisfied for $\Gamma$ at $v$.  Indeed, if an eigenfunction
  of $\Gamma$ vanishes at $v$, it automatically becomes the
  eigenfunction of $\Gamma_0$, still vanishing at $v$, which is a
  contradiction.
\end{proof}

\section{An application: connectedness of the secular manifolds}
\label{sec:secular}

In this section we will deal only with graphs with NK or
Dirichlet vertex conditions.  For such a graph $\Gamma$ it is possible
to find the eigenvalues $\lambda = k^2$ as the solutions of the
equation
\begin{equation}
  \label{eq:secdet}
  F_\Gamma(k) := C \det\left(e^{-ikL/2} I - e^{ikL/2} S\right) = 0,
\end{equation}
where all matrices have dimension $2E$ with $E$ being the number of
edges of the graph; they should be thought as operating on vectors
indexed by the directed edges of the graph (each edge corresponds to
two directed edges).  The matrix $I$ is the identity matrix, $L$ is
the diagonal matrix populated with the edge lengths and $S$ is a
unitary matrix with real entries of known form \cite{KotSmi_ap99} (the
precise form is irrelevant to our discussion).  The constant $C$ can
be chosen so that $F_\Gamma$ is real for real $k$.

Each length appears in the matrix $L$ twice: once for each direction
of the edge.  As a consequence, the diagonal matrix $e^{ikL}$ has two
entries $e^{ik \ell_e}$ for each edge $e$.  Substituting $k \ell_e$ with the
torus variables $\kappa_e \in [0,2\pi)$, we get the function
$\Phi_\Gamma(\kappa_1, \ldots, \kappa_E)$ such that
\begin{equation*}
  \Phi_\Gamma(k \ell_1, \ldots, k \ell_E) = F_\Gamma(k).
\end{equation*}
The solutions $\vec{\kappa}$ of $\Phi_\Gamma(\vec{\kappa}) = 0$ form
an algebraic subvariety $\Sigma_\Gamma$ of the torus $\mathbb{T}^E$.
We call $\Sigma_\Gamma$ the \emph{secular manifold} of the graph
$\Gamma$.  The study of $\Sigma_\Gamma$ as a tool of understanding
eigenvalues of a quantum graphs was pioneered by Barra and Gaspard
\cite{BarGas_jsp00}. 

It has been conjectured by Colin de Verdi\`ere \cite{CdV_ahp15} that
the secular manifold is irreducible if the graph has a symmetry which
is preserved under any change of edge lengths.  It can be shown (see
\cite{CdV_ahp15} for a partial proof) that in this case, the graph is
an interval, a circle, a mandarin \cite{BanBerWey_jmp15} (also called
``pumpkin graph'' by other authors \cite{Ken+_ahp16}) or has some
loops.  It is also conjectured that in this case the set of the
non-smooth points has co-dimension 2 with respect to the manifold
$\Sigma_\Gamma$.

In this section we prove a related result for a family of quantum
graphs.  We start with some terminology from \cite{CdV_ahp15} (whose
term for $\Sigma_\Gamma$ is ``determinant manifold'').  A point of
$\Sigma_\Gamma$ is \emph{smooth} if the differential of $\Phi_\Gamma$
at this point is non-zero.  A point $\vec{\kappa} \neq 0$ is smooth if
and only if $1$ is a non-degenerate eigenvalue of the graph $\Gamma$
with edge lengths set to $\ell_e = \kappa_e$ or, more generally, if
$\lambda=k^2$ is an eigenvalue of $\Gamma$ with lengths $\ell_e$ such
that $\vec{\kappa} = k \ell_e \mod 2\pi$.  In what follows we will
omit the ``modulo $2\pi$'' from the description of points on the torus,
to keep the notation compact.

\begin{thm}
  \label{thm:connected}
  Let the graph $\Gamma$ have no loops and have a vertex of degree
  one.  Then the set of smooth points of $\Sigma_\Gamma$ has two
  connected components.
\end{thm}

\begin{figure}
  \centering
  \includegraphics{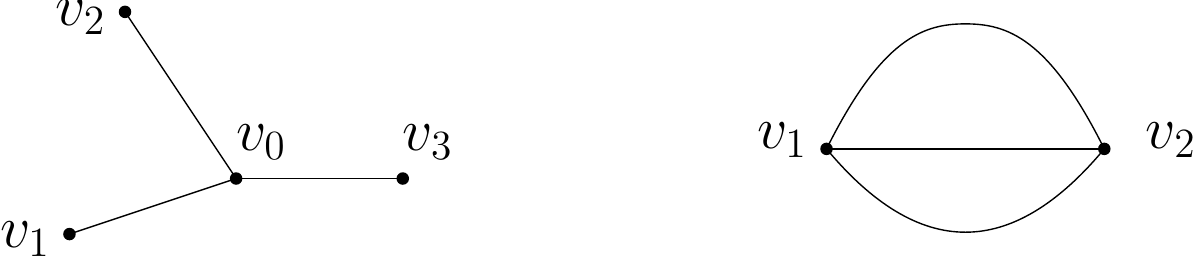}
  \caption{A star graph with three edges and a mandarin graph with three edges.}
  \label{fig:star_and_mandarin}
\end{figure}

\begin{figure}
  \centering
  {
    \includegraphics[scale=0.5]{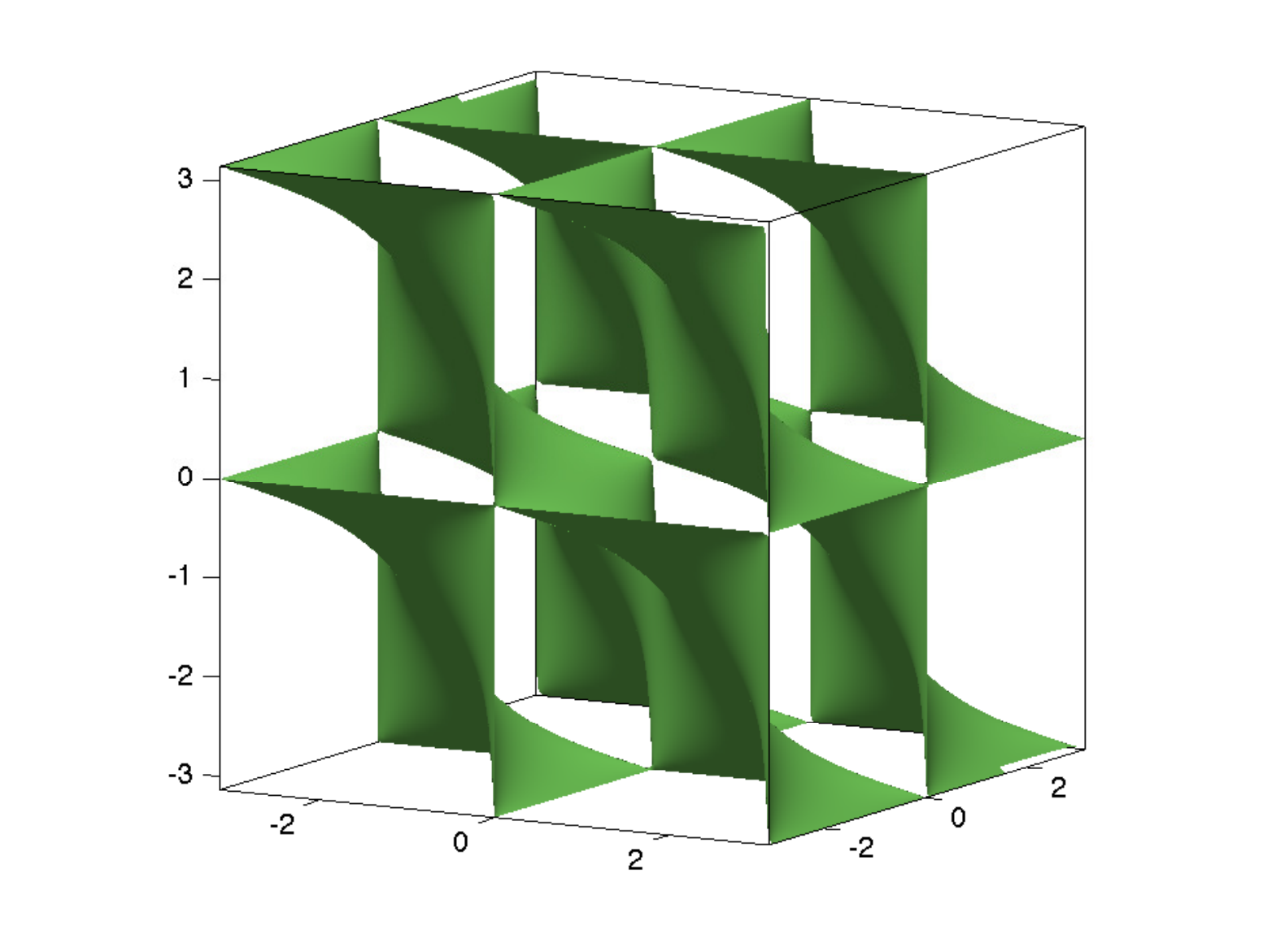}
    \includegraphics[scale=0.5]{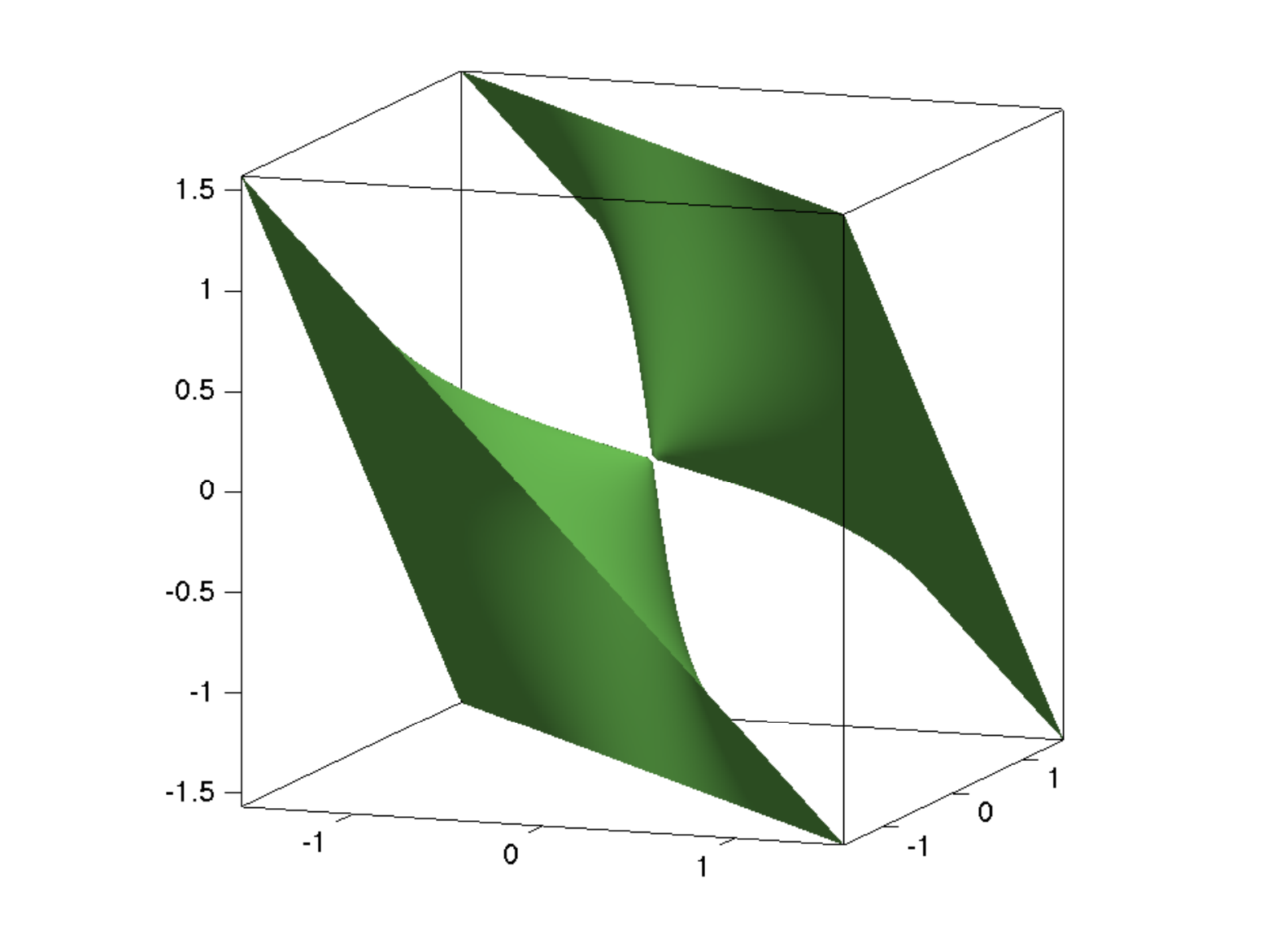}
  }
  \caption{Secular manifold of the star graph shown over
    $[-\pi,\pi]^3$ (left).  Of the four sheets visible, the first and
    third are parts of the same sheet by the torus periodicity; same
    for the sheets two and four.  A detail of the plot over
    $[-\pi/2,\pi/2]^3$ is shown on the right.}
  \label{fig:star_sec_manif}
\end{figure}

\begin{example}
  Consider a star graph with three edges $(v_0,v_1)$, $(v_0,v_2)$ and
  $(v_0,v_3)$ (see Fig.~\ref{fig:star_and_mandarin}) with NK
  condition at the central vertex $v_0$ and Dirichlet conditions at
  the leaves $v_1$, $v_2$, $v_3$.  There are three edges and therefore
  three torus variables.  It can be shown
  \cite{BarGas_jsp00,BerKea_jpa99} that the secular function of such a
  graph has the form
  \begin{equation}
    \label{eq:secdet_3starD}
    \Phi_{\mathrm{star,D}}(\vec\kappa) =
    \sum_{j=1}^3 \sin(\kappa_j)\sin(\kappa_{j+1})\cos(\kappa_{j+2}),
  \end{equation}
  where indices of $\kappa$ are taken modulo 3.  The secular manifold
  for the graph is shown in Figure~\ref{fig:star_sec_manif}.  There
  are four sheets visible but they match pairwise under the torus
  periodicity.  The two sheets touch each other through the conical
  points of non-smoothness.

  Similar pictures result if we consider the star graph with Neumann
  conditions at the leaves, which results in the secular determinant 
  \begin{equation}
    \label{eq:secdet_3starN}
    \Phi_{\mathrm{star,N}}(\vec\kappa) =
    \sum_{j=1}^3 \cos(\kappa_j)\cos(\kappa_{j+1})\sin(\kappa_{j+2})
   = \Phi_{\mathrm{star,D}}(\vec\kappa - \pi/2).
  \end{equation}
\end{example}

\begin{proof}[Proof of Theorem~\ref{thm:connected}]
  Choose the edge lengths in such a way that
  \begin{enumerate}
  \item \label{item:simple_nonzero} the eigenvalue spectrum of
    $\Gamma$ is simple and eigenfunctions do not vanish at vertices,
  \item \label{item:incommensurate} the edge lengths are rationally
    independent.
  \end{enumerate}
  Denote the vector of the edge lengths by $\vec\ell_0$.

  Condition (\ref{item:incommensurate}) implies that the flow $k \mapsto
  k\vec\ell_0$ is ergodic on the torus, and its intersections with the
  secular manifold are dense in it.  The closure of the odd-numbered
  intersections (within the set of smooth points of $\Sigma_\Gamma$)
  forms one component and the even-numbered intersections, the other.
  We will prove that they are mutually disjoint and connected.

  It is known (see, for example, \cite{Ban_ptrsa14} or
  \cite{CdV_ahp15}) that the gradient of $\Phi_\Gamma$ has either all
  non-negative components or all non-positive components.  Since $k_n
  = \sqrt{\lambda_n}$ are simple roots of the real-valued function
  $F_\Gamma(k) = \Phi_\Gamma(k\vec\ell_0)$, the derivatives
  $F'_\Gamma(k_n)$ alternate in sign and therefore the gradient is
  non-negative on one of the components we defined and non-positive on
  the other.  At a point of intersection of the two components, the
  gradient must vanish which contradicts the definition of the
  components.  

  Let now $\vec\kappa_1$ and $\vec\kappa_2$ be the two points on
  the same component (without loss of generality, take the component
  of even-numbered eigenvalues).  Surround them by small open
  neighbourhoods $U_1$ and $U_2$, such that $U_j \cap \Sigma_\Gamma$
  are connected and contain only smooth points of the same component.
  From definition of the components we can find two eigenvalues
  $\lambda_{2n_1} = k_1^2$ and $\lambda_{2n_2} = k_2^2$, $n_1<n_2$,
  such that $k_j\vec\ell_0 \in U_j \cap \Sigma_\Gamma$, $j=1,2$.  We
  now need to show that the points $k_1\vec\ell_0$ and $k_2\vec\ell_0$
  can be connected by a path on $\Sigma_\Gamma$ which does not pass
  through any singular points.

  Denote by $v$ the vertex of degree 1 and by $e_1$ the edge leading to it.
  The vertex condition at $v$ can be written as 
  \begin{equation}
    \label{eq:extended_delta}
    \cos(\theta/2) \frac{df}{dx_{e_1}}(v) = \sin(\theta/2) f(v),
  \end{equation}
  with $\theta = \theta_v$ equal to $0$ for Neumann and $\pi$ for
  Dirichlet.  Now we start with the eigenvalue $\lambda = k_2^2$ and
  continue it analytically by changing $\theta$.  According to
  \cite[Thm 6.1]{BerKuc_incol12} (see also \cite[Thm
  3.1.13]{BerKuc_graphs}), the eigenvalue $\lambda(\theta)$ is an
  analytic function unless there exists an eigenfunction of the graph
  $\Gamma$ which satisfies both Dirichlet and Neumann conditions (and
  thus any other $\delta$-type conditions) at the vertex $v$.  Such an
  eigenfunction would have to be identically zero on the edge $e_1$,
  which we ruled out in condition (\ref{item:simple_nonzero}) above.
  Therefore, $\lambda(\theta)$ is analytic and passes through every
  eigenvalue of $\Gamma$ at the points $\theta = 2\pi n + \theta_v$,
  $n\in\Z$, $n \geq n_0$.  In particular, for $n= \tilde{n} :=
  2(n_1-n_2) < 0$ we will have $\lambda(\theta) = k_1^2$.

  We will now map this $\lambda$-path, parametrized by $\theta$
  decreasing from $\theta_v$ to $2\pi \tilde{n} + \theta_v$, to a path
  on the secular manifold $\Sigma_\Gamma$.  Starting with an
  eigenfunction $f$ on $\Gamma$ with the vertex condition at $v$
  specified by $\theta$, we prolong the edge $e_1$ to have the length
  \begin{equation}
    \label{eq:loose_edge_length}
    \ell_{e_1}(\theta) := \ell_{0,e_1} 
    + \frac{\theta_v - \theta}{2\sqrt{\lambda(\theta)}}.
  \end{equation}
  A direct calculation shows that the eigenfunction $f$ continued as a
  sine wave with the same amplitude past the old location of the
  vertex $v$ will satisfy condition (\ref{eq:extended_delta}) with
  $\theta=\theta_v$ at the new location of $v$.

  Define the vector-function $\vec\ell(\theta)$ by using
  (\ref{eq:loose_edge_length}) for the component corresponding to
  $e_1$ and keeping all other components equal to the
  corresponding components of $\vec\ell_0$.  The above discussion
  shows that the point $\vec\kappa(\theta) = \sqrt{\lambda(\theta)}
  \vec\ell(\theta)$ will remain on the secular manifold for all
  $\theta$ between $\theta_v$ and $2\pi \tilde{n} + \theta_v$ and will
  pass only through points of multiplicity $1$, which are exactly the
  points of smoothness.  This path connects the points $k_1\vec\ell_0$
  and $k_2\vec\ell_0$ as required.  The seeming mismatch of lengths at
  $\theta = 2\pi \tilde{n} + \theta_v$ is due to the modular
  arithmetic on the torus; here we use the fact that $\tilde{n}$ is
  even.
\end{proof}

\begin{remark}
  \label{rem:bridge}
  An identical theorem can be proved for a graph with a bridge, i.e.\
  an edge whose removal disconnects the graph.  The method of proof is
  the same with a point an the bridge being the location where the
  variable $\delta$-type condition is introduced and the resulting
  eigenfunction is related to the eigenfunction of the original graph
  with a longer edge.  Whether the same result holds for graphs
  without such edges (such as the tetrahedron graph --- the complete
  graph on 4 vertices) is still unknown.
\end{remark}

\begin{figure}
  \centering
  {
    \includegraphics[scale=0.5]{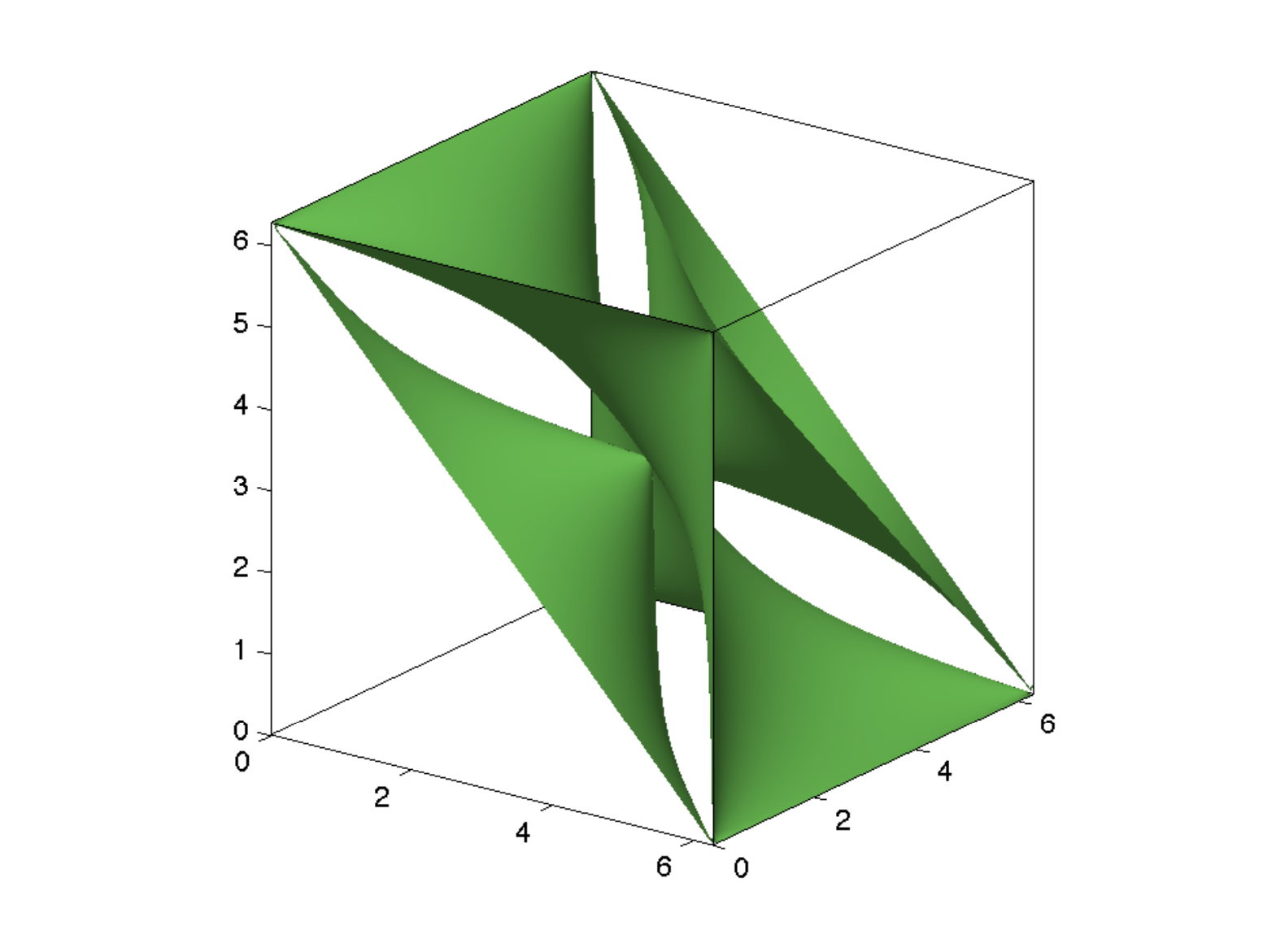}
    \includegraphics[scale=0.5]{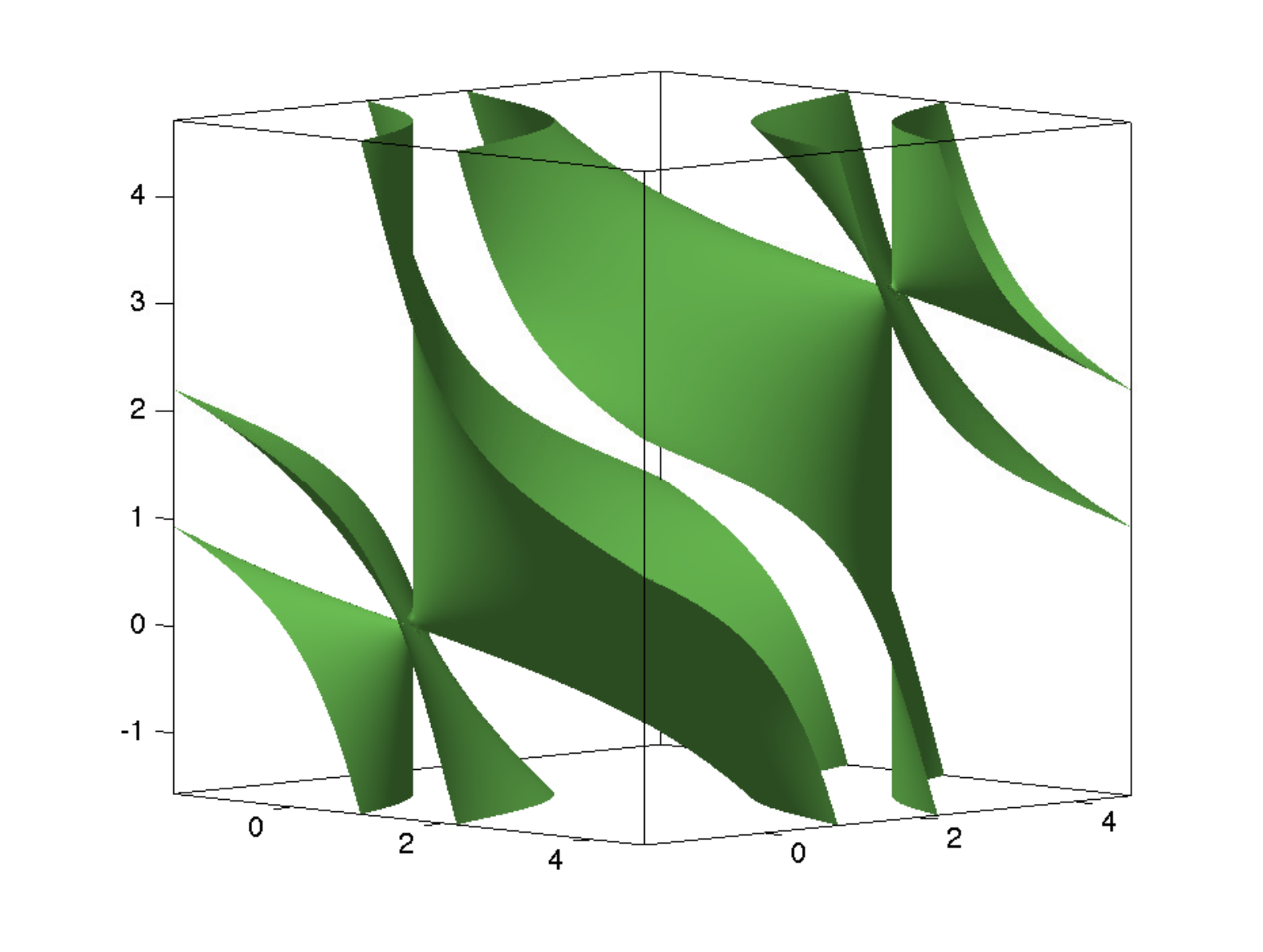}
  }
  \caption{Secular manifold of the mandarin graph shown over
    $[0,2\pi]^3$ (left) and, to provide another perspective, over
    $[-\pi/2, 3\pi/2]^3$ (right). 
  }
  \label{fig:mandarin_sec_manif}
\end{figure}

\begin{example}
  The following example shows that there is no direct link between
  reducibility of the secular manifold and its connectedness.  For the
  mandarin graph with three edges (see
  Fig.~\ref{fig:star_and_mandarin}), the secular determinant
  decomposes into a product of the secular determinants of Dirichlet
  and Neumann stars due to the reflection symmetry,
  \begin{equation*}
    \Phi_{\mathrm{mandarin}}(\vec\kappa) 
    = \Phi_{\mathrm{star,D}}(\vec\kappa/2) \cdot 
    \Phi_{\mathrm{star,N}}(\vec\kappa/2),
  \end{equation*}
  see equations (\ref{eq:secdet_3starD}) and
  (\ref{eq:secdet_3starN}).  Note that while the conditions of
  Theorem~\ref{thm:connected} (or Remark~\ref{rem:bridge}) are not
  satisfied, the secular manifold still has two connected components,
  see Fig.~\ref{fig:mandarin_sec_manif}.
\end{example}

\section*{Acknowledgment}

Both authors were partially supported by the NSF grant DMS-1410657.
We thank the referee for numerous improving suggestions and for
helping us identify and close a significant gap in the proof.

\bibliographystyle{plain}
\bibliography{bk_bibl}

\end{document}